\documentclass[letterpaper, 10 pt, conference]{ieeeconf}

\IEEEoverridecommandlockouts
\usepackage{cite}
\usepackage{amsmath,amssymb,amsfonts}
\usepackage{algorithmic}
\usepackage{graphicx}
\usepackage{textcomp}
\usepackage{xcolor}
\usepackage{braket}
\usepackage{mathtools}
\usepackage{pgfplots}
\usepackage{tikz}
\pgfplotsset{compat=1.18}

\newtheorem{theorem}{Theorem}
\newtheorem{corollary}{Corollary}

\title{\LARGE \bf Perturbation Analysis of Maximal Quantum Leakage}

\author{Zijia Zhao, Shuixin Xiao, and Farhad Farokhi
\thanks{The authors are with the Department of Electrical and Electronic Engineering at the University of Melbourne, Australia. This research was supported by the Commonwealth through an Australian Government Research Training Program Scholarship.}
\thanks{E-mails:\{zhaozj,shuixin.xiao,farhad.farokhi\}@unimelb.edu.au}
}

\begin{document}
\maketitle
\thispagestyle{empty}
\pagestyle{empty}

\begin{abstract}
Maximal quantum leakage (MQL) is a worst-case information leakage measure that quantifies an adversary's inference advantage gained from accessing quantum encoding of classical data with arbitrary measurements. While MQL admits an exact characterization for a given ensemble of quantum states, its robustness to implementation imperfections has not been systematically studied. In this paper, we analyze the sensitivity of maximal quantum leakage under perturbations of the quantum encoding. We establish a continuity bound in terms of the trace distance between ideal and perturbed quantum states, and show, via an example, that this bound is attainable. We further derive fidelity-based and relative-entropy-based sufficient conditions for bounding the variation of maximal quantum leakage, and illustrate numerically that these conditions can be loose.
\end{abstract}


\section{Introduction}

Information leakage through a computing or communication system has been the topic of both classical and quantum information theory. Even when a primary channel is designed to be secure, side-channel observations may allow an adversary to infer sensitive information about an underlying secret random variable. Classical measures, such as mutual information or channel capacity~\cite{calmon2012privacy, nekouei2019information, tanaka2017directed, wiese2018secure}, while fundamental, often fail to quantify worst-case adversarial behaviors.

To address this limitation, maximal leakage was introduced as an operational measure that quantifies the worst-case multiplicative increase in an adversary's optimal guessing probability after observing the system output\cite{saeidian2023, issa2020operational}. Unlike average information measures, such as mutual information, maximal leakage directly characterizes the strongest possible inference advantage over all decision strategies, making it well suited for side-channel analysis and privacy-sensitive system design~\cite{issa2020operational, wu2020, issa2016maximal, wu2020optimal}. 

As quantum technologies advance, classical data are increasingly encoded into quantum states, and adversaries may interact with these encodings through arbitrary quantum measurements~\cite{farokhi2023cdc, hirche2023qdp}. This motivates the notion of maximal quantum leakage (MQL), which extends maximal leakage to the quantum setting by allowing the adversary to perform arbitrary positive operator-valued measurements. The study~\cite{farokhi2024} provides an exact single-letter characterization of MQL for a given ensemble of quantum states, yielding a clear operational interpretation in terms of adversarial guessing advantage. More computable bounds have also been developed to shed light on the relationship between quantum maximal leakage and the distinguishability of encodings~\cite{farokhi2024barycentric}. Variants of quantum maximal leakage is shown to be a powerful technique in analyzing security of quantum key distribution mechanisms~\cite{farokhi2024measuring}.

The exact characterization of quantum maximal leakage alone does not address the robustness of leakage guaranties under perturbations of the underlying quantum states. In practice, quantum encodings are subject to noise, implementation errors, and unintended subsystem interactions. Knowing how maximal quantum leakage varies under such perturbations is therefore essential to assess leakage risks in realistic quantum systems~\cite{wilde2017, audenaert2007chernoff}.

A substantial body of work~\cite{winter2016continuity,audenaert2025continuity,audenaert2005relative} derives explicit bounds controlling how certain entropic and divergence-based quantities can vary under perturbations of quantum states, typically in settings where such quantities are well-behaved. These continuity-type bounds are developed for standard entropy and relative-entropy-based quantities, but they do not directly extend to maximal quantum leakage, which is defined via a worst-case optimization over measurements. Furthermore, a class of such continuity bounds based on relative entropy require extra assumptions on the eigenvalues of the density operator that can be restrictive in general.

In this work, we analyze the sensitivity of maximal quantum leakage to perturbations of quantum states. Using trace distance, we derive an explicit upper bound on the change in maximal quantum leakage and show that this bound is tight in the worst case by constructing encodings that attain it. We further obtain leakage bounds under fidelity-based and relative-entropy-based perturbations via inequalities relating these quantities to trace distance. Together, these results provide a quantitative robustness analysis of maximal quantum leakage under physically meaningful perturbations. We also present numerical examples demonstrating that the bounds relating to fidelity and relative-entropy perturbations can be loose. 

The rest of the paper is organized as follows. Section~\ref{sec:prem} present  preliminary results on quantum systems, distance measures, and quantum maximal leakage. Section~\ref{sec:bound} present  the derived continuity bounds and discuss their tightness in various settings. Finally, Section~\ref{sec:conc} concludes the paper.

\section{Preliminaries}
\label{sec:prem}

\subsection{Notation}
Let $\mathcal{H}$ be a finite-dimensional Hilbert space with dimension $d$, where $\mathcal{L}(\mathcal{H})$ denotes the set of all linear operators on $\mathcal{H}$. A pure state of a quantum system is represented by a unit vector $|\psi\rangle \in \mathcal{H}$ with $\langle \psi | \psi \rangle = 1$ where \(\langle\psi| := |\psi\rangle^\dagger \). A density operator is a positive semidefinite operator $\rho \in \mathcal{L}(\mathcal{H})$ with $\mathrm{Tr}(\rho)=1$. The set of all density operators on $\mathcal{H}$ is denoted by $\mathcal{D}(\mathcal{H})$. A density operator $\rho$ is said to be pure if $\rho = \ket{\psi}\bra{\psi}$ for some pure state $\ket{\psi} \in \mathcal{H}$; otherwise, $\rho$ is mixed. In general, any density operator admits a convex decomposition into pure states, i.e., $\rho = \sum_i p_i \ket{\psi_i}\bra{\psi_i}$, $p_i \ge 0,\;\; \sum_i p_i = 1$. An ensemble is a collection $\mathcal{E} = \{ p_X(x), \rho_x \}_{x \in \mathcal{X}}$, where $\{p_X(x)\}_{x\in\mathcal{X}}$ is a probability distribution over a finite alphabet $\mathcal{X}$ and $\rho_x \in \mathcal{D}(\mathcal{H})$ are density operators,
representing a random preparation of quantum states. A positive operator-valued measure (POVM) on $\mathcal{H}$ is a collection of operators $\{F_y\}_{y \in \mathcal{Y}} \subset \mathcal{L}(\mathcal{H})$ such that $F_y \ge 0$ for all $y \in \mathcal{Y}$ and $\sum_{y \in \mathcal{Y}} F_y = I$. By the Born's rule, when a POVM $\{F_y\}$ is applied to a quantum state $\rho$, the probability of obtaining outcome $y$ is given by $\Pr(y|\rho) = \operatorname{Tr}(\rho F_y)$. All logarithms are in base 2.

\subsection{Distance Measures Between Quantum States}

The trace norm between two density operators $\rho$ and $\sigma$ is defined as $\|\rho-\sigma\|_1 \, := \, \operatorname{Tr}(\sqrt{(\rho-\sigma)^\dagger(\rho-\sigma)})$. The fidelity between $\rho$ and $\sigma$ is defined as $F(\rho,\sigma) \;:=\; \big\|\sqrt{\rho}\sqrt{\sigma}\big\|_1^2$. The quantum relative entropy is defined as $D(\rho\|\sigma) := \operatorname{Tr}[\rho(\log \rho - \log \sigma)]$ when $\operatorname{supp}(\rho) \subseteq \operatorname{supp}(\sigma)$, and $D(\rho\|\sigma) = +\infty$ otherwise.

\subsection{Information Leakage and Maximal Quantum Leakage}

Let $X$ be a discrete random variable taking values in a finite alphabet $\mathcal{X}$, representing sensitive information to be protected. In an information leakage setting, an adversary aims to infer $X$, or a subset of $\mathcal{X}$ that contains it, by observing the output of a system that processes or encodes $X$. Information leakage measures quantify the extent to which access to such observations improves the adversary's ability to guess $X$ or a function of $X$.

In this paper, the random variable $X$ is encoded into quantum states. Specifically, each realization $x \in \mathcal{X}$ is mapped to a density operator $\rho_x$ on a Hilbert space $\mathcal{H}$, resulting in an ensemble $\mathcal{E}_\rho = \{p_X(x), \rho_x\}_{x \in \mathcal{X}}$. An adversary is allowed to perform an arbitrary quantum measurement on the state and use the measurement outcome to infer information about $X$.

Among various information leakage measures, maximal quantum leakage captures the worst-case multiplicative increase in an adversary’s guessing probability due to access to the measurement outcomes. 

Following \cite{farokhi2024}, the maximal quantum leakage associated with the ensemble $\mathcal{E}_\rho=\{p_X(x),\rho_x\}_{x\in\mathcal{X}}$ is given by
\[
\mathcal{Q}_\rho \;:=\; \sup_{\{F_y\}} \; \sup_{ Z,\, \hat{Z}} \log_2 \left( \frac{\mathbb{P}\!\left\{ Z = \hat{Z} \right\}} {\max_{z \in \mathcal{Z}} \mathbb{P}\!\left\{ Z = z \right\}} \right),
\]
where the inner supremum is taken over all random variables $Z$ and $\hat{Z}$ with arbitrary supports $\mathcal{Z}$ where $\hat{Z}$ is any estimate of $Z$ based on measurement outcome, while the outer supremum is taken over all POVMs $\{F_y\}$ on $\mathcal{H}$. Moreover, it follows from \cite[Theorem~1]{farokhi2024} that
\[
2^{\mathcal{Q}_\rho} = \sup_{\{F_y\}} \sum_{y} \max_{x \in \mathcal{X}} \operatorname{Tr}(\rho_x F_y).
\]

Operationally, maximal quantum leakage quantifies the maximum advantage that an adversary can obtain, in the worst case over all randomized functions of $X$, by observing the outcomes of a quantum measurement compared to having no access to the encoded quantum system.

\subsection{Perturbation Setting and Objective}

In many practical scenarios, the quantum encoding of a sensitive random variable may not be implemented exactly as intended. Instead of the ideal ensemble $\mathcal{E}_\rho = \{p_X(x), \rho_x\}_{x \in \mathcal{X}}$, the actual physical system may produce a perturbed ensemble $\mathcal{E}_{\rho'} = \{p_X(x), \rho'_x\}_{x \in \mathcal{X}}$, where each state $\rho'_x$ deviates slightly from the target state $\rho_x$. Such deviations may arise from imperfect isolation between subsystems, hardware noise, or unintended interactions within a quantum device.

The objective is to study the effect of such perturbations on maximal quantum leakage. Specifically, given two ensembles $\mathcal{E}_\rho$ and $\mathcal{E}_{\rho'}$ defined on the same alphabet
$\mathcal{X}$ and Hilbert space $\mathcal{H}$, we aim to quantify how the maximal quantum leakage changes when $\rho_x$ is replaced by $\rho'_x$ for each $x \in \mathcal{X}$. Our primary object is to bound the difference $\big| 2^{\mathcal{Q}_\rho} - 2^{\mathcal{Q}_{\rho'}} \big|$, which compares the worst-case adversarial advantage induced by the ideal and perturbed quantum encoding.

\section{Main Results}
\label{sec:bound}

We now present our main results on the stability of maximal quantum leakage under perturbations of the quantum encoding. We first establish a continuity bound in terms of the trace distance
between the ideal and perturbed states. Then we provide bounds based on fidelity and relative entropy.

\subsection{Trace-Distance Perturbation Bound}

The following theorem establishes a Lipschitz-type continuity bound for maximal quantum leakage with respect to the trace distance.

\begin{theorem} \label{tho:1}
    For a finite discrete random variable $X$ with distribution $p_X(\cdot)$, let $ \mathcal{E}_{\rho} = \{p_X(x), \rho_x\}_{x \in \mathcal{X}}$ and $ \mathcal{E}_{\rho'} = \{p_X(x), \rho'_x\}_{x \in \mathcal{X}}$ be ensembles such that for any  $x \in \mathcal{X}$ and $\epsilon > 0$, if $\| \rho_x - \rho'_x \|_1 \leq \epsilon$, then
    \[
    \left| 2^{\mathcal{Q}_{\rho}} - 2^{\mathcal{Q}_{\rho'}} \right| \leq  \frac{\epsilon}{2} \cdot \min\{|\mathcal{X}|, d\}.
    \]
\end{theorem}

\begin{proof}
    Suppose a POVM \(\{F_{y_0}\}_{y_0 \in \mathcal{Y}_0}\) maximizes 
    \[\mathcal{Q}_{\rho}= \sup_{\{F_y\}}\log_2\left(\sum_{y \in \mathcal{Y}}\max_{x \in \mathcal{X}} \operatorname{Tr}(\rho_x F_y) \right).\] Then $2^{\mathcal{Q}_{\rho}} - 2^{\mathcal{Q}_{\rho'}}$
    \begin{align*}
        =& \sum_{y_0 \in \mathcal{Y}_0}\max_{x \in \mathcal{X}} \operatorname{Tr}(\rho_x F_{y_0}) - \sup_{\{F_{y'}\}}\sum_{y' \in \mathcal{Y}'}\max_{x \in \mathcal{X}} \operatorname{Tr}(\rho'_x F_{y'}) \\
        \leq& \sum_{y_0 \in \mathcal{Y}_0}\max_{x \in \mathcal{X}} \operatorname{Tr}(\rho_x F_{y_0})  - \sum_{y_0 \in \mathcal{Y}_0}\max_{x \in \mathcal{X}} \operatorname{Tr}(\rho'_x F_{y_0})
    \end{align*}

    Now we only need to prove for any set of POVM \(\{F_y\}\), $\mathcal{D} :=  \sum_{y\in \mathcal{Y}}\max_{x \in \mathcal{X} }\operatorname{Tr}(\rho_x F_y) - \sum_{y\in \mathcal{Y}}\max_{x \in \mathcal{X} }\operatorname{Tr}(\rho'_x F_y)$ is less than or equal to $\frac{\epsilon}{2} \cdot \min\{|\mathcal{X}|, d\}$.
    Suppose \(s(y) \in \operatorname*{arg\,max}_{x \in \mathcal{X}} \operatorname{Tr}(\rho_x F_y) \), then
    \begin{align*}
        \mathcal{D} &\leq \sum_{y}\operatorname{Tr}(\rho_{s(y)} F_y) - \sum_{y}\operatorname{Tr}(\rho_{s(y)}' F_y)\\
        &= \sum_{y}\operatorname{Tr}((\rho_{s(y)} - \rho'_{s(y)}) F_y)\\
        &= \sum_{x} \operatorname{Tr}((\rho_x - \rho_x')E_x),
    \end{align*}
    where $E_x := \sum_{y \in s^{-1}(x)}F_y$. Since \(\rho_x, \rho_x'\) are Hermitian matrices, we have \(\rho_x - \rho_x'\) is also Hermitian and it has a spectral decomposition \(\rho_x - \rho_x' = \sum_{x_i} \lambda_{x_i}  \ket{v_{x_i}}\!\bra{v_{x_i}}\). By the Jordan-Hahn decomposition decomposition of the Hermitian operator $\rho_x - \rho'_x$\cite{billingsley1995}, it admits a decomposition into positive and negative parts 
    $\rho_x - \rho'_x = \Delta_{x+} - \Delta_{x-}$, 
    where 
    \begin{align*}
        \Delta_{x+} :=& \sum_{\lambda_{x_i} > 0} \lambda_{x_i} \ket{v_{x_i}}\!\bra{v_{x_i}}\geq 0,\\
        \Delta_{x-} :=& -\sum_{\lambda_{x_i} < 0} \lambda_{x_i} \ket{v_{x_i}}\!\bra{v_{x_i}}\geq 0.
    \end{align*}
    Then
    \begin{align*}
        & \operatorname{Tr}(\Delta_{x+}) + \operatorname{Tr}(\Delta_{x-}) =  \sum_{x_i} |\lambda_{x_i}| = \| \rho_x - \rho'_x \|_1.
    \end{align*}
    Furthermore, 
    \begin{align*}
        \operatorname{Tr}(\Delta_{x+}) - \operatorname{Tr}(\Delta_{x-}) = \operatorname{Tr}(\rho_x - \rho'_x) = 0.
    \end{align*}
    Combining these identities results in
    \begin{align*}
        \operatorname{Tr}(\Delta_{x+}) = \operatorname{Tr}(\Delta_{x-}) = \frac{1}{2}\| \rho_x - \rho'_x \|_1 \leq \frac{\epsilon}{2}.
    \end{align*}
    Based on the above decomposition, we have
    \begin{align*}
        \mathcal{D} &\leq \sum_{x} \operatorname{Tr}((\rho_x - \rho_x')E_x)\\
        &= \sum_{x} \left(\operatorname{Tr}(\Delta_{x+}E_{x}) - \operatorname{Tr}(\Delta_{x-}E_{x})\right)\\
        &\leq \sum_{x} \operatorname{Tr}(\Delta_{x+}E_{x})\\
        &= \sum_{x} \operatorname{Tr} \left(\left( \sum_{\lambda_{x_i} > 0} \lambda_{x_i} \ket{v_{x_i}}\!\bra{v_{x_i}}\right) E_{x}\right).\\
    \end{align*}
Note that 
\begin{align*}
    \sum_{\lambda_{x_i} > 0} \lambda_{x_i} \ket{v_{x_i}}\!\bra{v_{x_i}} 
    &\leq \sum_{\lambda_{x_i} > 0} \left( \sum_{\lambda_{x_i} > 0} \lambda_{x_i} \right)  \ket{v_{x_i}} \bra{v_{x_i}} \\
    &= \operatorname{Tr}({\Delta_{x+}})P_x,
\end{align*}
where \(P_x := \sum_{i > 0} \ket{v_i}\!\bra{v_i}\) is the projector onto the positive eigen-space of \(\rho_x - \rho'_x\). Therefore, we have
    \begin{align*}
        \mathcal{D} &\leq \sum_{x} \operatorname{Tr} \left( \operatorname{Tr}({\Delta_{x+}}) P_{x} E_{x}\right)\\
        &= \sum_{x} \operatorname{Tr}({\Delta_{x+}}) \operatorname{Tr} \left(P_{x} E_{x}\right)\\
        &\leq \frac{\epsilon}{2} \sum_{x} \operatorname{Tr} \left(P_{x} E_{x}\right)\\
        &= \frac{\epsilon}{2} \operatorname{Tr} \left(\sum_{x} \left( P_{x} \sum_{y \in s^{-1}(x)}F_y \right) \right)\\
        &\leq \frac{\epsilon}{2} \operatorname{Tr} \left( \left( \sum_{y \in \mathcal{Y}}F_y \right) \left(\sum_{x} P_{x} \right) \right)\\
        &\leq \frac{\epsilon}{2} \operatorname{Tr} \left(\sum_{x} P_{x} \right).
    \end{align*} 
    Since the operator \(\sum_{x} P_{x}\) has rank at most \(\min\{|\mathcal{X}|, d\}\), we have \(\mathcal{D}\leq \frac{\epsilon}{2} \cdot \min\{|\mathcal{X}|, d\}.\) Thus 
    \[
    2^{\mathcal{Q}_{\rho}} - 2^{\mathcal{Q}_{\rho'}} \leq \mathcal{D}\leq \frac{\epsilon}{2} \cdot \min\{|\mathcal{X}|, d\}.
    \]    
    By symmetry, we also have \(2^{\mathcal{Q}_{\rho'}} - 2^{\mathcal{Q}_{\rho}} \leq  \frac{\epsilon}{2} \cdot \min\{|\mathcal{X}|, d\},\) which finishes the proof.
\end{proof}

In what follows, by constructing an example, we demonstrate that the bound in Theorem~\ref{tho:1} is tight in the worst case. Suppose \(\mathcal{X} = \{1, \cdots, |\mathcal{X}|\}\), for any integer $x \in \mathcal{X}$, let $\rho_x = \ket{x}\!\bra{x}$ with $|\mathcal{X}| \leq d$. We begin by evaluating $2^{\mathcal{Q}_{\rho}}$ for the ensemble $\mathcal{E}_\rho = \{p_X(x), \rho_x\}_{x\in\mathcal{X}}$ with arbitrary underlying distribution. The proposition 2 of \cite{farokhi2024} shows that $2^{\mathcal{Q}_{\rho}} \leq |\mathcal{X}|$. Now consider the POVM $\{F_y\}_{y \in \mathcal{Y}}$ such that
\[
F_y := |y\rangle\!\langle y| \ \text{ for } y\in\{1,\dots,d\}, \quad \sum_{y=1}^d F_y = I_d.
\]
With this choice of POVM, the maximal quantum leakage of $\mathcal{E}_\rho$ evaluates to
\begin{align*}
    2^{\mathcal{Q}_{\rho}} 
    &= \sum_{y=1}^d \max_x \operatorname{Tr}(\rho_x F_y) \\
    &= \sum_{y=1}^{|\mathcal{X}|} 1 + \sum_{y=|\mathcal{X}|+1}^d 0 \\
    &=  |\mathcal{X}|.
\end{align*}
For \(\epsilon \in (0,1)\), let $\rho_x' = \left(1-\frac{\epsilon}{2}\right)\ket{x}\!\bra{x} + \frac{\epsilon}{2}\ket{x+1}\!\bra{x+1}$, where the addition is taken modulo $|\mathcal{X}|$, and ensemble $\mathcal{E}_{\rho'} = \{p_X(x), \rho'_x\}_{x\in\mathcal{X}}$. Under this construction, the trace norm satisfies \(\| \rho_x - \rho_x' \|_1 = \epsilon\), and it holds that
\begin{align*}
\operatorname{Tr}(\rho_x' E_x) 
&\le \max\!\left\{\left(1-\frac{\epsilon}{2}\right)(E_x)_{xx}, \frac{\epsilon}{2}(E_x)_{x+1,x+1}\right\} \\
&\le 1-\frac{\epsilon}{2},
\end{align*}
where the indices are taken modulo $|\mathcal{X}|$. Therefore,
\begin{align*}
    \sum_y \max_x \operatorname{Tr}(\rho_x' F_y) 
    &= \sum_x \operatorname{Tr}(\rho_x' E_x) \\
    &\leq |\mathcal{X}| \cdot \left(1-\frac{\epsilon}{2}\right),
\end{align*}
which implies $2^{\mathcal{Q}_{\rho'}} \le |\mathcal{X}|\!\left(1-\frac{\epsilon}{2}\right)$. With the same POVM $\{F_y\}_{y \in \mathcal{Y}}$, $\sum_{y=1}^d \max_x \operatorname{Tr}(\rho_x' F_y) = \sum_{y=1}^{|\mathcal{X}|} \left(1-\frac{\epsilon}{2}\right) + \sum_{y=|\mathcal{X}|+1}^d 0= |\mathcal{X}| \cdot \left(1-\frac{\epsilon}{2}\right)$. Thus \(2^{\mathcal{Q}_{\rho}} = |\mathcal{X}|\), \(2^{\mathcal{Q}_{\rho'}} = \left(1-\frac{\epsilon}{2}\right) |\mathcal{X}|\) and \(2^{\mathcal{Q}_{\rho}} - 2^{\mathcal{Q}_{\rho'}} = \frac{\epsilon}{2}\cdot|\mathcal{X}|\). This example saturates the upper bound for Theorem 1. Noting that maximal leakage is invariant under unitary operations~\cite{farokhi2024}, this example can be used to constructed infinitely many situations in which the upper bound is attained.

\subsection{Fidelity-Based Leakage}
We next derive perturbation bounds for quantum maximal  leakage based on quantum fidelity.
\begin{corollary}
For a finite discrete random variable $X$ with distribution function $p_X(\cdot)$, let $\mathcal{E}_{\rho} = \{p_X(x), \rho_x\}_{x \in \mathcal{X}}$ and $\mathcal{E}_{\rho'} = \{p_X(x), \rho'_x\}_{x \in \mathcal{X}}$ be ensembles such that \(1-\epsilon < F(\rho_x,\rho'_x) \leq 1\) for any $x \in \mathcal{X}$ some $\epsilon > 0$. Then
\[
\left| 2^{\mathcal{Q}_{\rho}} - 2^{\mathcal{Q}_{\rho'}} \right| \leq \sqrt{\epsilon} \cdot \min\{|\mathcal{X}|, d\}.
\]
\end{corollary}
\begin{proof}
By the Fuchs--van de Graaf inequalities \cite[Theorem~9.3.1]{wilde2017}
\[
1-\sqrt{F(\rho, \sigma)} \leq \frac{1}{2}\|\rho-\sigma\|_1 \leq \sqrt{1-F(\rho, \sigma)}.
\]
and the assumption $F(\rho_x,\rho'_x) > 1-\epsilon$ for any $x \in \mathcal{X}$, it follows that 
\[
\frac{1}{2}\|\rho_x-\rho'_x\|_1 \le \sqrt{1-F(\rho_x,\rho'_x)} < \sqrt{\epsilon}.
\]
Hence $\|\rho_x-\rho'_x\|_1 \le 2\sqrt{\epsilon}$.
Applying Theorem 1 yields
\begin{align*}
    \left| 2^{\mathcal{Q}_{\rho}} - 2^{\mathcal{Q}_{\rho'}} \right| 
    &\le \frac{\|\rho_x-\rho'_x\|_1}{2} \cdot \min\{|\mathcal{X}|, d\} \\
    &\le \sqrt{\epsilon} \cdot \min\{|\mathcal{X}|, d\}.
\end{align*}
\end{proof}

\subsection{Relative-Entropy-Based Leakage}
A sufficient condition for leakage control can also be derived in terms of quantum relative entropy.

\begin{corollary}
For a finite discrete random variable $X$ with distribution function
$p_X(\cdot)$, let $\mathcal{E}_{\rho} = \{p_X(x), \rho_x\}_{x \in \mathcal{X}}$ and $\mathcal{E}_{\rho'} = \{p_X(x), \rho'_x\}_{x \in \mathcal{X}}$ be ensembles such that $D(\rho_x \| \rho'_x) \le \epsilon$ for all $x \in \mathcal{X}$ and some $\epsilon > 0$. Then
\[
\left| 2^{\mathcal{Q}_{\rho}} - 2^{\mathcal{Q}_{\rho'}} \right|
\le  \frac{\sqrt{2 \ln 2 \cdot \epsilon}}{2} \cdot \min\{|\mathcal{X}|, d\}.
\]
\end{corollary}

\begin{proof}
According to the quantum Pinsker inequality in \cite[Theorem~11.9.1]{wilde2017},
\[
\|\rho_x-\rho'_x\|_1 \le \sqrt{2 \ln 2 \cdot D(\rho_x\|\rho_x')} \le \sqrt{2 \ln 2 \cdot \epsilon}.
\]
Applying Theorem 1 yields
\[
\left| 2^{\mathcal{Q}_{\rho}} - 2^{\mathcal{Q}_{\rho'}} \right| \le \frac{\sqrt{2 \ln 2 \cdot \epsilon}}{2} \cdot \min\{|\mathcal{X}|, d\}.
\]
\end{proof}

\subsection{Numerical Illustration}

We provide numerical results to investigate the potential looseness of stability bounds based on both fidelity and quantum relative entropy. In our simulations, we fix the Hilbert space dimension to $d=3$ and generate a large pool of random density operators sampled from the induced Ginibre ensemble \cite{ginibre1965}. For each trial, a base ensemble $\mathcal{E}=\{\rho_1,\rho_2,\rho_3\}$ is constructed by independently sampling three states from the same distribution.

Given a perturbation radius $\varepsilon>0$, we associate with each $\rho_i$ a candidate pool consisting of all states $\sigma$ satisfying either $1-F(\rho_i,\sigma)\le \varepsilon$ in the fidelity-based setting, or $D(\rho_i\|\sigma) \le \varepsilon$ in the relative-entropy-based setting. Perturbed ensembles $\mathcal{E}'=\{\rho'_1,\rho'_2,\rho'_3\}$ are then formed by independently sampling each $\rho'_i$ from the pool corresponding to $\rho_i$. For both $\mathcal{E}$ and $\mathcal{E}'$, the maximal quantum leakage is computed numerically by optimizing over POVMs using the iterative projected-gradient algorithm in \cite{farokhi2024}. Repeating this experiment over many independent base ensembles and sweeping over $\varepsilon$, we record the worst-case deviation in $2^{\mathrm{MQL}} = \max_{\mathcal{E}_{\rho'}}\bigl|2^{\mathcal{Q}_{\rho'}} - 2^{\mathcal{Q}_\rho}\bigr|$ induced by admissible perturbations, thereby empirically characterizing the looseness of fidelity-based and relative-entropy-based stability bounds.

\begin{figure}[t]
  \centering
  \includegraphics[width=0.95\linewidth]{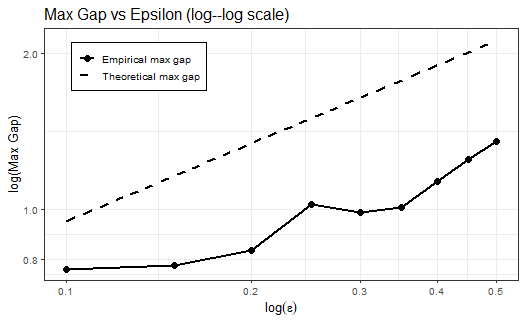}
  \caption{Maximum deviation of $2^{\mathrm{MQL}}$ under fidelity-based perturbations as a function of the radius $\varepsilon$.}
  \label{fig:fidelity}
\end{figure}

\begin{figure}[t]
  \centering
  \includegraphics[width=0.95\linewidth]{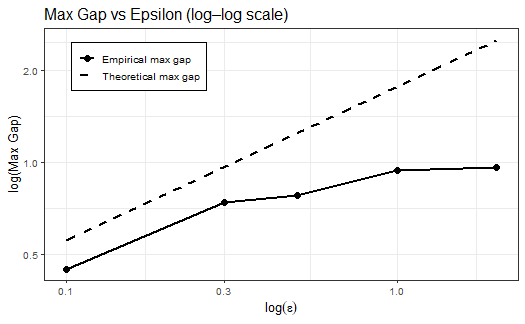}
  \caption{Maximum deviation of $2^{\mathrm{MQL}}$ under relative-entropy-based perturbations as a function of the radius $\varepsilon$.}
  \label{fig:relentropy}
\end{figure}

The resulting log--log plots of the maximum deviation in $2^{\mathrm{MQL}}$ as a function of the perturbation radius $\varepsilon$ are shown in Fig.~\ref{fig:fidelity} and Fig.~\ref{fig:relentropy} for the fidelity-based and relative-entropy-based settings, respectively. In each figure, we additionally plot the corresponding theoretical upper bound on the variation of $2^{\mathrm{MQL}}$, as predicted by the analytical stability results.

In both cases, the empirical maximum gap exhibits an approximately linear trend on the log--log scale over a nontrivial range of $\varepsilon$, suggesting a power-law-type dependence of $2^{\mathrm{MQL}}$ on the perturbation radius. However, a substantial gap is observed between the numerically computed deviations and the theoretical upper bounds, indicating a potential looseness of the bounds in practice. Taken together, these observations demonstrate that while both fidelity- and relative-entropy-based bounds are valid in theory, their degrees of looseness differ markedly.

\section{Conclusion}
\label{sec:conc}

In this paper, we studied the stability of maximal quantum leakage under perturbations of the quantum encoding. We established a trace-distance continuity bound that quantifies the variation of maximal quantum leakage between an ideal and a perturbed ensemble, and showed that this bound is tight in the worst case. We further derived fidelity- and relative-entropy-based sufficient conditions and illustrated, via numerical results, that these conditions can be generally loose. Our results provide a quantitative understanding of the robustness of maximal quantum leakage to implementation imperfections. An interesting direction for future work is to identify tighter bounds under additional structural assumptions on the quantum encoding or to extend the analysis to other operational leakage measures.

\bibliographystyle{ieeetr}
\bibliography{ref}

\end{document}